\documentclass [11pt,a4paper]{article}

\usepackage{amssymb}
\usepackage{enumerate,verbatim}
\usepackage{amsmath,amssymb}
\usepackage{amsfonts,amsthm}
\usepackage{wasysym}
\usepackage{statmath}
\usepackage{bm}

\usepackage{epsfig,lpic,tikz}
\usepackage{statmath}
\usepackage{graphicx}
\usepackage{subfigure}
\usepackage[utf8]{inputenc}

\usepackage{pdfsync}
\usepackage{xcolor,hyperref}
\usepackage[T1]{fontenc}
\usepackage{mathtools}

\hypersetup{
    colorlinks=true,
    linkcolor=blue,
    filecolor=magenta,
    urlcolor=blue,%cyan,
    pdftitle={Overleaf Example},
    pdfpagemode=FullScreen,
    }

\newcounter{teoremaganso}
\newcounter{appendix}

\newcounter{coryganso}

\renewenvironment{abstract}{\small\quotation\noindent
 {\bfseries \abstractname .}}{\endquotation \par}

\catcode`\@=11

\def\resetthefootnote{\renewcommand{\thefootnote}{\@arabic\c@footnote} }
\def\@principiremex#1{\trivlist
 \item[\hskip \labelsep{\bfseries #1\ \thetheo.}]\ignorespaces}
\def\opar@principiremex#1[#2]{\trivlist
 \item[\hskip \labelsep{\bfseries #1\ \thetheo\ (#2).}]\ignorespaces}

\newcommand{\newTHEOremrom}[2]{\newenvironment{#1}{\refstepcounter{theo}\@ifnextchar[{\opar@principiremex{#2}}
{\@principiremex{#2}}}{\qedB\endtrivlist}} \catcode`\@=12
\DeclareMathSymbol{\square}{\mathord}{AMSa}{"03}
\newcommand{\qedB}{\nopagebreak\hspace*{\fill}$\square$\par}
\newcommand{\Qed}{\nopagebreak\hspace*{\fill}{\vrule width6pt height6pt depth0pt}\par}

\newtheorem*{thmA}{Theorem A}
\newtheorem*{thmB}{Theorem B}

\theoremstyle{plain}

\newtheorem{theorem}{Theorem}[section]
\newtheorem{lemma}[theorem]{Lemma}

\newtheorem{definition}{Definition}
\newtheorem{remark}{Remark}

\newsavebox{\savepar}

\title{\bf A Convex-Geometric Framework for Fully Phase-Locked States in the Finite Kuramoto Model}

\author{Antonio Garijo$^1$, Sergio G\'omez$^{1,2}$ and Alex Arenas$^{1,2,3}$\thanks{Corresponding author: \texttt{alexandre.arenas@urv.cat}}
\\*[.1truecm]
{\small \textsl{$^{1}$ Departament d'Enginyeria Inform{\`a}tica i Matem{\`a}tiques,}}
\\*[-.05truecm]
{\small \textsl{Universitat Rovira i Virgili, 43007 Tarragona, Spain}}
\\*[.1truecm]
{\small \textsl{$^{2}$ ComSCIAM, Universitat Rovira i Virgili, 43007 Tarragona, Spain}}
\\*[.1truecm]
{\small \textsl{$^{3}$ Complexity Hub Vienna, Metternichgasse 8, 1030 Wien, Austria}}
}

\date{}
\begin{document}

\maketitle
\begin{abstract}
We study the finite-size Kuramoto model of all-to-all coupled phase oscillators with heterogeneous natural frequencies and characterize the minimal coupling strength required for the existence of a fully phase-locked equilibrium (in a co-rotating frame). To remove the degeneracy due to uniform phase shifts, we move to a reduced co-rotating frame and assess stability through the Jacobian of the reduced system: a fully phase-locked state is stable when this Jacobian is negative definite. This defines a stability region in the phase space. The Kuramoto vector field maps this region to a convex set in frequency space, so a fully-locked state at coupling $K$ exists exactly when the rescaled frequency vector $\hat{\bfomega}/K$ lies inside that convex image. The critical coupling $K_{\ell}$ is defined as the smallest coupling strength for which a fully phase-locked equilibrium exists; geometrically, it corresponds to the first intersection of the ray $t\hat{\bfomega}$ with the boundary of this convex set.
Building on this convex-geometric structure, we construct an explicit polytope from analytically computable boundary points of the stability region, providing a closed-form upper bound $K_b \ge K_{\ell}$. The bound is exact for frequencies aligned with polytope vertices and offers a fully explicit outer approximation for general frequency vectors. While not uniformly sharp in a quantitative sense, this construction exposes the underlying geometry of stable fully phase-locking solutions.
These results provide a practical use the convex-geometric structure underlying stable fully-locked states in the Kuramoto model.
\end{abstract}

\section{Introduction}
Synchronization of interacting oscillatory units is a ubiquitous collective phenomenon, appearing across the natural sciences and engineering in contexts as diverse as circadian rhythms and neuronal activity, chemical oscillations, and power-grid dynamics. A common modeling paradigm describes each unit by a phase variable and encodes interactions through coupling terms determined by the underlying network, giving rise to a rich repertoire of behaviors, including phase locking, partial synchrony, and complex macroscopic transitions \cite{Pikovsky_Book, Strogatz_survey}. Within this broad class of phase-oscillator models, the Kuramoto model introduced in 1975 \cite{Kuramoto_Model} has become a central and widely used framework because it captures the competition between intrinsic frequency heterogeneity and attractive coupling while remaining analytically tractable in many regimes. Its mean-field formulation and network generalizations have stimulated extensive theoretical work on synchronization thresholds and bifurcation structure, as well as applications in physics, biology and technology \cite{Acebron_Review, Dorfler_survey}. From a finite-size perspective, a complementary line of work provides explicit synchronization conditions and sharp inequalities that relate frequency heterogeneity to the existence of phase-locked states \cite{CliqueDeVille}. It is important to distinguish between two mathematically distinct problems: (i) the loss of stability of the incoherent state and the associated emergence of macroscopic coherence, typically analyzed via linearization in the thermodynamic limit; and (ii) the existence and stability of fully phase-locked equilibria in finite populations, which constitutes a nonlinear fixed-point problem for the phase differences. The present work addresses the latter.

In the following, we introduce the finite-size Kuramoto dynamics and the order parameter that quantifies phase coherence, and then develop a geometric approach to bound the coupling required for stable phase-locked equilibria.

\section{Model, Reduction, and Stability Notions}

The Kuramoto model of coupled oscillators  is a mathematical framework describing synchronization processes. This model, introduced in  1975 by Kuramoto  (\cite{Kuramoto_Model}), represents a continuous dynamical system governed by a first order differential equation defined in the $(n+1)-$dimensional torus $\mathbb T^{n+1}$.  Using the compact notation  $\bfomega = (\omega_1, \ldots, \omega_n, \omega_{n+1})$, $\bftheta = (\theta_1, \ldots, \theta_n, \theta_{n+1})$,  $\bm{f}=(f_1, \ldots, f_n, f_{n+1})$ and $K \geq 0$ a real parameter, the Kuramoto model writes as
\begin{equation}\label{eq:main}
\frac{d\bftheta}{dt}  = \bfomega + K \bm{f}(\bftheta)
\end{equation}
\noindent where ${\bftheta}$ belongs to the $(n+1)-$torus
$$ \mathbb T^{n+1} = \mathbb S^1 \times  \cdots \times \mathbb S^1 =  \left \{ (\theta_1, \ldots, \theta_n, \theta_{n+1}) \, | \, \theta_i \in (-\pi, \pi],\ 1 \leq i \leq n+1 \right \}
 $$
 and $\bm{f}:\mathbb T^{n+1} \to \mathbb T^{n+1}$ writes as $\bm{f} = (f_1, \ldots, f_n, f_{n+1})$ where
\begin{equation}\label{eq:def_f}
    f_i(\theta_1, \ldots, \theta_n, \theta_{n+1})=\displaystyle  \sum_{j=1, \, j \neq i }^{n+1} \sin{(\theta_j - \theta_i)} \quad \text{for} \quad 1 \leq i \leq n+1.
 \end{equation}

The \emph{synchronization order parameter} $R$ measures the degree of synchronization of the system and is given by
\begin{equation}\label{eq:R}
R e^{ i \psi} = \frac{1}{n+1} \displaystyle \sum_{j=1}^{n+1}  e^{i \theta_j}.
\end{equation}
By definition, $R \in [0,1]$, and values of $R$ close to 1 represent a synchronization state, while values of $R$ close to 0 represent the non-synchronization one. Moreover, $\psi$ represents the average phase of the system.

By the nature of the equation driving system~\eqref{eq:main}, stable solutions are not isolated. More precisely, if $\bftheta^*$ is an equilibrium point and given any angle $\alpha \in (-\pi, \pi]$ then  $\bftheta^* + \bfalpha= (\theta^*_1+\alpha, \ldots, \theta^*_n+\alpha, \theta^*_{n+1} + \alpha)$ it is also an equilibrium point. We can identify all these solutions writing $[\bftheta^*]= \{\bfalpha +  \bftheta^*  \, | \, \alpha \in (-\pi,\pi]\}$. This fact is also visible in the eigenvalues of the Jacobian matrix $D\bm{f}(\bftheta)$ since $  (1, \ldots, 1)$ is an eigenvector of  $D\bm{f}(\bftheta)$ with eigenvalue $\lambda=0$ for all $\bftheta \in \mathbb T^{n+1}$. Thus, $\lambda=0$ is always in the spectrum of $D\bm{f}(\bftheta)$.  There are several ways to avoid this situation. We adopt the following one: reducing to $n$ the number of variables. More precisely, we introduce the new variables
\[
\hat{\theta}_i = \theta_i - \theta_{n+1} \qquad  1 \leq i \leq n.
\]
In this rotation frame, the ODE that governs the Kuramoto model~\eqref{eq:main} is transformed into
\begin{equation}\label{eq:main1}
\frac{d\hat{\bftheta}}{dt}= \hat{\bfomega} + K   \bm{F} ( \hat{\bftheta})
 \end{equation}
where $ \hat{\omega}_i = \omega_i - \omega_{n+1}$  for $1 \leq i \leq n$, and $\bm{F} ( \hat{\bftheta}) = (\hat{f}_1, \ldots, \hat{f}_n)$ is defined by
\begin{equation}\label{eq:f_hat}
\hat{f}_i (\hat{\theta}_1, \ldots, \hat{\theta}_n)  = - \sin(\hat{\theta}_i)
- \displaystyle \sum_{j=1}^{n} \sin(\hat{\theta}_j)
+ \displaystyle \sum_{j=1, \, j \neq i}^n \sin{(\hat{\theta}_j - \hat{\theta}_i)}.
\end{equation}

We also denote by $D\bm{F}$ the Jacobian matrix  given by  $D\bm{F}(\hat{\theta}_1, \ldots, \hat{\theta}_n))=\left(\frac{\partial\hat{f}_i}{\partial \hat{\theta}_k}\right)_{1\leq i,k \leq n}$ where
\begin{equation}\label{eq:derivativef}
\frac{\partial \hat{f}_i}{\partial \hat{\theta}_k} =
\left\{
\begin{array}{ll}
\cos(\hat{\theta}_k-\hat{\theta}_i) - \cos(\hat{\theta}_k),  &  k \neq i,
\\
-2\cos(\hat{\theta}_i) -\displaystyle\sum_{j=1,\, j\neq i} ^n\cos(\hat{\theta}_j-\hat{\theta}_i), \,  & k=i.
\end{array}
\right.
\end{equation}
In particular, at the origin ${\bf0}=(0, \ldots, 0)$ we have that
\begin{equation}\label{eq:DF_0}
D\bm{F}(\bf0)= \left(
\begin{array}{cccccc}
-(n+1) & 0 & 0&\cdots & 0& 0 \\
0 & -(n+1) & 0& \cdots &0 &0 \\
0 & 0 & -(n+1)& \cdots & 0 &0 \\
\vdots & \vdots &  \vdots& \ddots & \vdots & \vdots \\
0 & 0 & 0 & \cdots  & -(n+1) & 0 \\
0 & 0 & 0 & \cdots  & 0 & -(n+1)
\end{array}
\right)
\end{equation}
\noindent is a diagonal matrix. Adopting this framework, we have that the determinant of $D\bm{F}(\bf0)$ is different from zero, thus eliminating the zero eigenvalue in the original formulation. However, the matrix $D\bm{F}(\hat{\bftheta})$ is no longer symmetric, but steal has real eigenvalues, as we see in the next lemma.

Since, from now on, we will just work with the reduced Kuramoto dynamics, it is useful to define a synchronization order parameter $\hat{R}$ based on $\hat{\bftheta}$ instead of $\bftheta$:
\begin{equation}\label{eq:R_red}
\hat{R} e^{i \hat{\psi}} = \frac{1}{n+1} \left[ 1+ \displaystyle \sum_{j=1}^{n}  e^{i \hat{\theta}_j}\right].
\end{equation}

\begin{lemma}\label{lem:DF}
Let $M=\left(
\begin{array}{cccc}
2 & 1 & \cdots &  1 \\
1 & 2 &  \cdots  &1 \\
\vdots &   \vdots& \ddots &  \vdots \\
1 & 1  & \cdots  & 2
\end{array}
\right)$ and $H(\hat{\bftheta})=  \frac{1}{2} \left| 1 + \displaystyle \sum_{j=1}^n e^{i \hat{\theta}_j} \right|^2$.
Then,

\begin{enumerate}[(a)]
    \item $\bm{F} (\hat{\bftheta}) = M  \, \nabla H(\hat{\bftheta}) $.
    \item $D\bm{F}(\hat{\bftheta}) = M  \, D^2H(\hat{\bftheta})$.
    \item $D\bm{F}(\hat{\bftheta})$ has real eigenvalues.
\end{enumerate}

\end{lemma}

\begin{proof}
It is easy to check that $H(\hat{\bftheta})=\frac{(n+1)^2}{2} \hat{R}^2(\hat{\bftheta})$. We first see that $\bm{F}(\hat{\bftheta}) = M \cdot \nabla H(\hat{\bftheta})$. Computing the derivatives of $H(\hat{\bftheta})$ we obtain

\begin{eqnarray}
\frac{\partial H}{\partial \theta_k} (\hat{\bftheta}) &= &
 -\sin\theta_k
\left(1+\displaystyle \sum_{j=1}^n \cos\theta_j \right)+  \cos(\theta_k) \left(\displaystyle\sum_{j=1}^n \sin\theta_j\right) \nonumber \\
& = & -\sin(\theta_k) + \displaystyle \sum_{j=1}^n [ \cos(\theta_k) \sin(\theta_j) -
\sin(\theta_k) \cos(\theta_j)] \nonumber \\
& = & -\sin(\theta_k) +  \displaystyle \sum_{j=1}^n \sin(\theta_j - \theta_k).
\end{eqnarray}
Moreover,
\[
\displaystyle \sum_{k=1}^{n} \frac{\partial H}{\partial \theta_k} (\hat{\bftheta}) = -\displaystyle \sum_{k=1}^n \sin(\hat{\theta}_k).
\]
So, from (5) we recall that $F= (\hat{f}_1, \ldots, \hat{f}_n)$, having
\[
\hat{f}_k(\hat{\bftheta}) = -\sin(\theta_k) +  \displaystyle \sum_{j=1}^n \sin(\theta_j - \theta_k) -
\displaystyle \sum_{k=1}^n \sin(\hat{\theta}_k) =
\frac{\partial H}{\partial \theta_k} (\hat{\bftheta}) +
\displaystyle \sum_{k=1}^{n} \frac{\partial H}{\partial \theta_k} (\hat{\bftheta}).
\]
Finally, we conclude that $F(\hat{\bftheta})=M \cdot \nabla H(\hat{\bftheta}).$

We observe that $M$ is a positive definite symmetric matrix with eigenvalues $1$ (with multiplicity $n-1$) and $n+1$ (with multiplicity 1).  Hence, $D\bm{F}(\hat{\bftheta}) = M  \cdot D^2H(\hat{\bftheta}).$ We denote by $M^{1/2}$ the symmetric square root matrix verifying $M^{1/2} M^{1/2}=M$.

In order to see that $D\bm{F}(\hat{\bftheta})$ has real eigenvalues we can see that it is similar to a symmetric matrix. More precisely, we claim that the matrices  $M  \cdot D^2H(\hat{\bftheta})$ and $M^{1/2} \cdot D^2H(\hat{\bftheta})  \cdot M^{1/2} $ are similar. To see the claim,  we notice that
\[
M \cdot D^2H(\hat{\bftheta}) = M^{1/2} M^{1/2} \cdot D^2H(\hat{\bftheta}) =  M^{1/2}  ( M^{1/2} \cdot D^2H(\hat{\bftheta}) \cdot M^{1/2} ) M^{-1/2}
\]
\noindent proving thus that $D\bm{F}(\hat{\bftheta})$ has real eigenvalues.

\end{proof}

\begin{remark}\label{remark:vaps_reals}

In this rotation frame \eqref{eq:main1}, the Jacobian matrix $D\bm{F}(\hat{\bftheta})$ is not symmetric. However, from Lemma~\ref{lem:DF}(c) we know that $D\bm{F}(\hat{\bftheta})$ has real eigenvalues since it is similar to a symmetric matrix.

\end{remark}

\subsection{Phase-locked equilibria and stability}\label{subsec:equilibria_stability}

An equilibrium (phase-locked) solution of~\eqref{eq:main1} is a point $\bftheta^*\in\mathbb{T}^n$ satisfying
\[
\hat{\bfomega}+K\,\bm{F}(\bftheta^*)=\bm{0}.
\]
Following \cite{Bronski_DeVille_Park}, we call this equilibrium \emph{stable} if the Jacobian $D\bm{F}(\bftheta^*)$ exhibits  negative eigenvalues. We denote this condition by $D\bm{F}(\bftheta^*) \prec 0$. As we mentioned above, the non symmetric matrix $D\bm{F}(\bftheta^*)$ has real eigenvalues (see Remark~\ref{remark:vaps_reals}).
%We say that $\hat{\bftheta}^*$ is an {\it equilibrium solution} of the system \eqref{eq:main1} as long as $\hat{\omega}_i + K   \bm{F} ( \hat{\bftheta}) =\bf0$. Following  \cite{Bronski_DeVille_Park} we introduce the notion of a {\it stable equilibrium solution} of the Kuramoto model \eqref{eq:main1}.
%\begin{defi}
%We say  that $\bftheta^*=(\theta^*_1, \ldots,\theta^*_n)$  is an stable solution of \eqref{eq:main1} if and only if
%\begin{enumerate}
%\item[(a)]  $\hat{\bfomega} + K \bm{F} (\hat{\bftheta}^*)=\bm{0}$,
%\item[(b)] $D\bm{F}(\hat{\bftheta}^*)$  is negative definite.
%\end{enumerate}
%\end{defi}
Notice that by introducing the rotating variables $\hat{\bftheta}=(\hat{\theta}_1, \ldots, \hat{\theta}_n)$ we have eliminated the translated solution and now the equilibrium solutions are isolated. Furthermore, the number of equilibrium solutions is finite; one can see this by rewriting the fixed-point equations as a polynomial system and invoking B\'ezout-type arguments (see \cite{Morse_Baillieul}),   and recent work further investigates the multiplicity and stability of these solutions in the finite-size Kuramoto model \cite{NumberSolutions}. Other relevant works on this subject provide interesting insights, see for instance   \cite{Hsiao2025equivalence,Hsiao2026Complexity}.

There are many introductory reviews on the state of the art of the Kuramoto models (see, among others, \cite{Dorfler_survey, Strogatz_survey, AlexSurvey}). We briefly recall some properties of this model.
We can distinguish between two different behaviors from the dynamical system point of view.
The first one corresponds to the case when $\hat{\bfomega}=\bf 0$. In this case, the Kuramoto model can be described as a potential system generated by a map $V:\mathbb T^n \mapsto \mathbb R$ and, basically, the flow is a downhill flow where the $\alpha$-limit and the $\omega$-limit of every orbit is an equilibrium point of the system. Furthermore, the stable equilibrium solutions of the system correspond to local minima of the map $V$. This case can be understood following the Morse Theory, mainly developed by J.~Milnor (\cite{MorseTheory_Milnor}). Moreover, $\bftheta^*=\bf 0$ is always a stable solution of the system, and the parameter $K$ has no effect on the system.

The second case corresponds to $\hat{\bfomega} \neq \bf 0$, where a critical coupling exists:
 \[
K_{\ell} \;=\; \inf \Bigl\{\, K>0 \;:\; \exists\, \boldsymbol{\theta}^* \in \mathbb T^n
\text{ stable equilibrium of~\eqref{eq:main1}} \Bigr\}.
\]
$K_{\ell}$ separates the case $K > K_{\ell}$, in which the system exhibits a stable equilibrium state, from the case $K < K_{\ell}$, in which no equilibrium solution exists. However, the Kuramoto model can present many different equilibria solution. In our approach we can follow, as a function of $K$, a selected principal branch of stable equilibrium. More precisely, using that $D\bm{F}({\bf0})$ is invertible we can construct a principal branch $K \mapsto \hat{\bftheta}(K)$, for sufficiently large $K$. Besides, we can continue this branch as long as the Jacobian
$D\bm{F}(\hat{\bftheta}(K))$ remains non-singular. We collect the main properties of this principal branch $K \mapsto\hat{\bftheta}(K)$ in Theorem~A. See Figure~\ref{fig:order_parameter} for an example of the map $ K \mapsto R(\hat{\bftheta}(K))$, where $R$ is the order parameter. Finally, the solution $\bftheta^*=\bf0$ is never an equilibrium solution of the system.

%Moreover, in this case, the Kuramoto model can exhibit a chaotic behavior (\cite{Chaos_Kura1,Chaos_Kura2}).

%The behavior of the Kuramoto model when $\hat{\bfomega} \neq \bf %0$ can be completely different from the downhill flow. Fixing the value of $\hat{\bfomega}$, it is possible to follow the stable equilibrium solution $\bftheta^*(K)$ depending on the $K$-parameter.

\begin{definition}
Let $\hat{\bfomega} \neq \bf0$. We consider the auxiliary map
\[
G(K,\hat{\bftheta}): = \hat{\bfomega} + K \bm{F}(\hat{\bftheta}), \qquad K > 0, \, \hat{ \bftheta} \in \mathbb T^n.
\]
Since $\bm{F}(\bf0)=\bf0$ and $D\bm{F}(\bf0)$ are invertible, there exist an open interval $(K_{br},\infty) \subset (0,\infty)$  and a unique $\mathcal C^1$ map
\[
\hat{\bftheta}_{pr}: (K_{br},\infty) \to \mathbb T^n
\]
verifying $G(K,\hat{\bftheta}_{pr}(K))=\bf0$ and $D\bm{F}(\hat{\bftheta}_{pr}(K)) \prec 0 ,$
and  $$\displaystyle \lim_{K \to \infty} \hat{\bftheta}_{pr}(K)=0.$$
This branch is called the {\it principal branch}. Moreover, we know that $K_{\ell} \leq K_{br}$.
\end{definition}

\begin{figure}[tbh]
    \centering
     \includegraphics[width=0.90\textwidth]{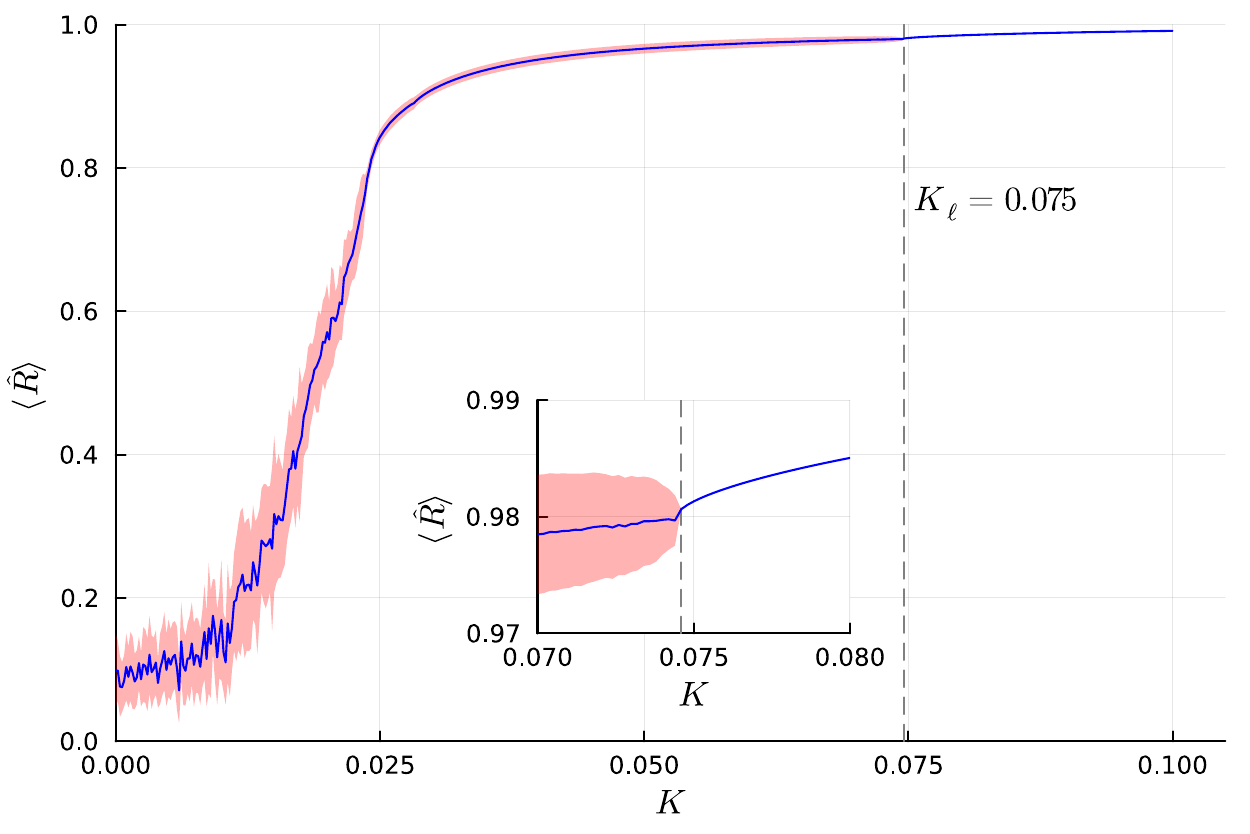}
           \caption{\small{
Dependence of the Kuramoto order parameter $\hat{R}$ on the coupling strength $K$ for a fixed nonzero frequency vector $\hat{\bfomega}$ with $n+1=100$. The frequencies are randomly drawn from a normal $N(0,1)$ distribution, setting $\omega_{n+1}=0$, adding an outlier at $\omega_1$ (equal to three times the largest of the absolute values of the other frequencies) to increase the heterogeneity. For each value of $K$, the system~\eqref{eq:main1} is solved using the 4th~order Runge-Kutta method up to $t_{\max}=200$. The last 20\% of the evolution is used to find the average $\langle \hat{R} \rangle$ and the standard deviation $\sigma_{\hat{R}}$ of the order parameter for each $K$. In this plot, we represent $\langle \hat{R} \rangle$ as a blue solid line, surrounded by a pale-red ribbon of width $\sigma_{\hat{R}}$. The vertical dashed line at $K=K_{\ell}$ marks the transition to the fully phased-locked transition, and the inset shows a zoom around this transition. Note that, in the fully-locked state $K>K_{\ell}$, the ribbon is almost absent, $\sigma_{\hat{R}}(K)\approx 0$, indicating the perfect locking of the phases due to the stability of the solution. The bound obtained for this example using Eq.~\eqref{eq:K_bound} is $K_{b}=0.905$.}}
      \label{fig:order_parameter}
\end{figure}

\begin{thmA} \label{theoremA}

Let $\hat{\bfomega} \neq \bf0$. The reduced  Kuramoto model
\[
\frac{d\hat{\bftheta}}{dt}= \hat{\bfomega} + K   \bm{F} ( \hat{\bftheta})
\]
\noindent satisfies the following:
\begin{enumerate}[(a)]
\item {\bf Nonexistence for small coupling.} There exists $K_0>0$ such that, for every $K \in (0,K_0)$,  the system has no equilibrium solutions.
\item {\bf Existence of a principal branch for large coupling.} There exists a principal branch
\[
\hat{\bftheta}_{pr}: (K_{br},\infty) \to \mathbb T^n
\]
defined on a maximal open interval $(K_{br},\infty) \subset (0,\infty)$, such that $$ \hat{\bfomega} + K \bm{F}(\hat{\bftheta}) ={\bf0}, \quad K \in (K_{br},\infty), $$
and $$\displaystyle \lim_{K \to \infty} \hat{\bftheta}_{pr}(K)=0.$$
\item {\bf Monotonicity of the order parameter.} Define
$$ R_{pr}(K):=\hat{R}( \hat{\bftheta}_{pr}(K)), \quad K > K_{br}.$$
Then $R_{pr}$ is a $\mathcal C^1$
 map and strictly increasing on $(K_{br},\infty)$, and
 $$
 \displaystyle \lim_{K \to \infty} R_{pr}(K)=1.
 $$
 \item {\bf Blow-up at the stability boundary.} It is verified that
 $$
\frac{d}{dK} R_{pr}(K) \to +\infty , \quad \text{ as }  K \to K_{br}^+.
 $$
 \end{enumerate}

%exhibits the following dichotomy
%\begin{enumerate}[(a)]
%\item for $K$ sufficiently small, the Kuramoto model has no equilibrium %solutions and,
%\item for $K$  sufficiently large, the Kuramoto model has a stable %equlibrium solution $\bftheta^*(K)$.
%\end{enumerate}
%Moreover, the order mapping
% \[
% \begin{array} {lccl}
% R: &  [K_{\ell} \, , +\infty) &  \mapsto  &  [0,1] \\
%  & K &\mapsto & R(\bftheta^*(K))
%  \end{array}
% \]
% is a  estrictly increasing map verifying $\displaystyle \lim_{K \to %\infty} R(\bftheta^*(K))=1$ and  $\displaystyle \lim_{K \to K_{\ell}^+} %R'(\bftheta^*(K))=+\infty$.
 \end{thmA}

% We now define the notion of a stable equilibrium solution following \cite{Bronski_DeVille_Park}.

We denote by $\Omega$ the set of frequencies $\hat{\bftheta} \in \mathbb T^n$  where all the  eigenvalues of $DF(\hat{\bftheta})$ are strictly negative, denoted by $D\bm{F}( \hat{\bftheta}) \prec 0 $. Thus,

\[
\Omega =  \{  \hat{\bftheta} \in \mathbb T^n \, | \,  D\bm{F}( \hat{\bftheta}) \prec 0   \}.
\]
We recall that the Jacobian matrix $D\bm{F}( \hat{\bftheta})$ has real eigenvalues (see Remark \ref{remark:vaps_reals}). The behavior of the order parameter $R$ depending on the coupling strength $K$ for a fixed value of $\hat{\bfomega}$ is depicted in Figure~\ref{fig:order_parameter}.

Hereafter, we fix a value of $\hat{\bfomega}$. We briefly describe the method to find an upper bound, denoted by $K_b$, of the critical parameter $K_{\ell}$. Thus, our method produces a fully explicit value $K_b$ satisfying $K_b \ge K_{br} \geq K_{\ell}$, obtained from a constructive outer approximation of the convex stability image.

\section{Convex Geometry and Main Bound}

By definition $\Omega$ is an open set. We first find $2n$ points on the boundary of $\Omega$ (see Lemma~\ref{lem:boundary_omega}). We denote these ``marked'' points  in $\partial \Omega$ by $p_i^{\pm}$ for $1 \leq i \leq n$.
We consider $\bm{F}(\Omega)$ as a set which is a convex set. See Figure~\ref{fig:Omega_F_Omega}. This property can be deduced from  \cite{Bronski_DeVille_Park}, see subsection \ref{subsec:Proof_B}.

\begin{remark}\label{rem:symmetries}
From the definition of $\Omega$ we can easily conclude that $\hat{\bftheta} \in \Omega$ if and only if $-\hat{\bftheta} \in \Omega$ since $D\bm{F} (\hat{\bftheta})=D\bm{F} (-\hat{\bftheta})$. In a similar way we have that $\hat{\bftheta} \in \bm{F}( \Omega) $ if and only if $-\hat{\bftheta} \in \bm{F}( \Omega)$. In this case, we notice that $\bm{F}(-\hat{\bftheta})=-\bm{F}(\hat{\bftheta})$.
\end{remark}

We have the following dichotomy:
\begin{enumerate}[(a)]
\item If $- \frac{\hat{\bfomega}}{K}$ belongs to  $\bm{F} (\Omega)$ then the Kuramoto model exhibits a stable solution.
\item If $-\frac{\hat{\bfomega}}{K}$  does not belong to $\bm{F} (\Omega)$ then the Kuramoto model does not exhibit a stable solution.
\end{enumerate}

\begin{figure}[tbp]
    \centering
     \includegraphics[width=0.9\textwidth]{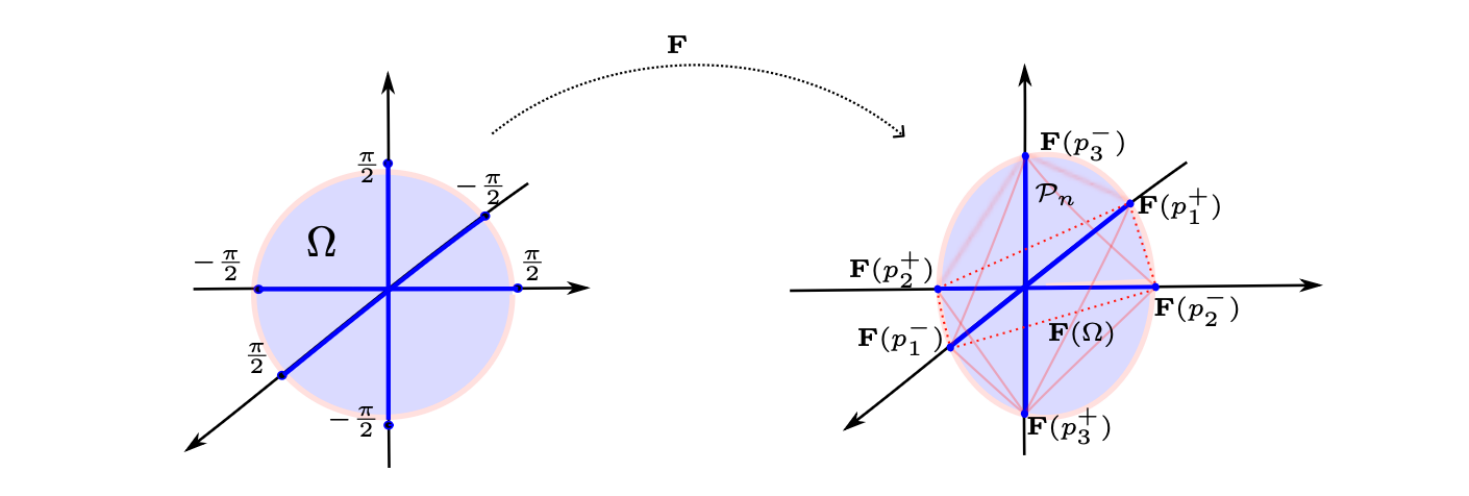}
 % \put(-267,47){\tiny $-\frac{\pi}{2}$}
 % \put(-240,50){\small $\Omega$}
 % \put(-190,47){\tiny $\frac{\pi}{2}$}
 % \put(-205,65){\tiny $-\frac{\pi}{2}$}
 % \put(-229,70){\tiny $\frac{\pi}{2}$}
 % \put(-235,10){\tiny $-\frac{\pi}{2}$}
 % \put(-255,25){\tiny $\frac{\pi}{2}$}
 % \put(-65,85){\tiny $\bm{F}(\Omega)$}
 % \put(-155,98){\tiny $\bm{F}$}
 % \put(-70,65){\tiny $\mathcal P_n$}
 % \put(-97,14){\tiny $\bm{F}(p_1^-)$}
 % \put(-55,12){\tiny $\bm{F}(p_2^-)$}
 % \put(-52,70){\tiny $\bm{F}(p_1^+)$}
 % \put(-107,70){\tiny $\bm{F}(p_2^+)$}
 % \put(-115,46){\tiny $\bm{F}(p_3^+)$}
 % \put(-47,36){\tiny $\bm{F}(p_3^-)$}
\caption{\small
Convex-geometric construction used to bound the critical coupling.
Left: the stability region $\Omega\subset\mathbb T^n$ (shaded) defined by $\mathrm{Spec}(D\bm{F}(\hat{\bftheta}))\subset(-\infty,0)$, with marked boundary points $p_i^\pm\in\partial\Omega$.
Right: its convex image $\bm{F}(\Omega)\subset\mathbb R^n$ (shaded) and the polytope $\mathcal P_n=\mathrm{conv}\{\bm{F}(p_i^\pm)\}$ (red), which provides an explicit outer approximation for locating the first intersection of the ray $t\hat{\bfomega}$ and hence an upper bound $K_b\ge K_{\ell}$.}
    \label{fig:Omega_F_Omega}
\end{figure}

Thus, $K_{br}$ is just the value where $- \frac{\hat{\bfomega}}{K}$ belongs to the boundary of $\bm{F} (\Omega)$. We know that $ \frac{\hat{\bfomega}}{K}$ belongs to $\bm{F} (\Omega)$ if and only if  $ -\frac{\hat{\bfomega}}{K}$ belongs to $\bm{F} (\Omega)$, see Remark \ref{rem:symmetries}. Related formulations of synchronization conditions in terms of functional inequalities and seminorms are developed in \cite{CliqueDeVille}.
Here, we construct a polytope $\mathcal P_n$ with $2n$ points in $\partial \bm{F}(\Omega)$.  Moreover, we can compute the faces of this polytope explicitly, see Lemma~\ref{lem:polytope_n}. We refer to \cite{ConvexPolytopesBranko} for background on convex polytopes and their face structure.
The geometry of the convex set $\bm{F}(\Omega)$ is closely related with the critical parameter $K_{\ell}$. The polytope $\mathcal P_n$ is the smallest convex set containing the marked points $\bm{F}(p_i^{\pm})$ for $1 \leq i \leq n$. Thus, considering the ray $t \hat{\bfomega}$ for $t >0$, we intersect first the boundary of the polytope $\mathcal P_n$ than the boundary of $\bm{F}(\Omega)$, see Figure~\ref{fig:Omega_F_Omega}. Finally, we are able to find the point $t_b$ such that $t_b\hat{\bfomega}$ belongs to the boundary of the polytope ${\mathcal P}_n$, and with it we can write $K_b=1/t_b$. In Theorem~B we show the analytical value of this upper bound $K_b$ of $K_{br}$ and thus an upper bound of the critical value $K_{\ell}$. This bound reflects the chosen outer approximation and is generally not expected to coincide with the exact boundary except in special geometric configurations.

\begin{thmB}
We consider the Kuramoto model~\eqref{eq:main1}  with $ \hat{\bfomega}=(\hat{\omega}_1, \ldots, \hat{\omega}_n) \neq \bf0.$  Then, the value

%Let $$\lambda_j =  \frac{1}{n+1} \left[ \hat{\omega}_j + \displaystyle  \sum_{k=1}^n \hat{\omega}_k\right ]  \text{  for }  1 \leq j \leq n.$$
%We denote by $i_1, \ldots, i_k$ the indexes such that $\lambda_j\leq0$ and $i_{k+1}, \ldots, i_n$ the indexes with $\lambda_j >0$.
\begin{equation}\label{eq:K_bound}
K_b = \frac{1}{n+1} \displaystyle \sum_{j=1}^n |\hat{\omega}_j|,
\end{equation}
is an upper bound of the critical parameter $K_{\ell}$.
\end{thmB}

% We remark that, assuming
% $ \displaystyle \sum_{k=1}^n \hat{\omega}_k=0$ as an extra hypothesis of Theorem~B, we obtain a simpler expression for the upper bound, namely $\displaystyle K_b= \frac{1}{n+1} \sum_{k=1}^n |\hat{\omega}_k|.$

 In Section~\ref{sec:Proofs} we prove the main results of this paper, Theorems~A and~B.

 \section{Proofs and Geometrical Analysis}\label{sec:Proofs}

 \subsection{Proof of Theorem A}

 %\begin{proof}[]
 We first prove statement (a). The non-existence of equilibrium solutions for $K$ small enough can be deduced by the boundedness of the maps $\hat{f}_i$~\eqref{eq:f_hat}. More precisely, defining
 $K_0 = \displaystyle \max_{ 1 \leq i \leq n}\left\{ \frac{|\hat{ \omega}_i|}{2n}\right\}>0$,
we claim that the Kuramoto model does not have any equilibrium solution for $K$ in the interval $(0,K_0)$. To prove this claim, we denote by $m$ the index where $K_0$ assumes its maximum. Since functions $f_i$ are sums of $2n$ sinus, we have
\[
 -2n \leq \hat{f}_m(\theta_1, \ldots, \theta_n) \leq 2n
\]
and thus
\[
\hat{\omega}_m - 2Kn \leq \hat{\omega}_m + K \hat{f}_m(\theta_1, \ldots, \theta_n) \leq \hat{\omega}_m + 2Kn.
\]
We firstly assume that $\hat{\omega}_m >0$ and thus $\hat{\omega}_m - 2K n >0$, since $K < K_0=\frac{\hat{\omega}_m}{2n}$, which implies that~\eqref{eq:main1} does not have an equilibrium solution. When $\hat{\omega}_m <0$ we have that $\hat{\omega}_m + 2K n<0$ since $K < K_0= \frac{-\hat{\omega}_m}{2n}$, which also implies the non existence of equilibrium solutions.

We secondly prove statement (b).  When $K$ is large enough the situation is just the converse, i.e. there always exists a stable equilibrium solution. Since $\bm{F}({\bf0})=\bf 0$ and $|D\bm{F}({\bf 0})|\neq 0$, and by the Open Mapping Theorem, there exists an open neighborhood $U$ of the origin where the map $\bm{F}: U \mapsto \bm{F}(U)$ is bijective. Thus, for $K$ sufficiently large, we always have a solution of $\bm{F}({\bftheta}^*)=-\frac{\hat{\bfomega}}{K}$, by just taking $K$ such that $-\frac{\hat{\bfomega}}{K}$ lies in $\bm{F}(U)$, which is a neighborhood of the origin. We denote this solution by $ {\bftheta}^*(K)$. This result basically implies that, for $K$ large enough, we have a stable solution ${\bftheta}^*(K)$ of the reduced Kuramoto model.

 The existence of a principal branch for large values of the parameter $K$ comes since we can follow a stable equilibrium solution via the Implicit Function Theorem. We define the auxiliary function $\bfG(K,\hat{\bftheta}) := \hat{\bfomega} + K \bm{F} (\hat{\bftheta})$. We consider
a sufficiently large  value $K_1$ such that $\bftheta^*$ is a stable solution of the Kuramoto model~\eqref{eq:main1}. In other words, we have a solution of the equation
$\bfG(K,\hat{\bftheta}) = \hat{\bfomega} + K \bm{F} (\hat{\bftheta})=\bf0$ at the point $K=K_1$ and $\hat{\bftheta}=\bftheta^*$.  Applying the Implicit Function Theorem there exists $K_{br}>0$ an a unique $\mathcal C^1$ map
\[
\hat{\bftheta}_{pr}: (K_{br}, \infty) \mapsto \mathbb T^n
\]
such that $\hat{\bftheta}_{pr}(K_1)=\bftheta^*$ and  $\bfG(K,\hat{\bftheta}_{pr}(K))=\bf0$  as long as the determinant $\frac{\partial \bfG}{\partial \bftheta} (\bftheta^*)=K D\bm{F}(\bftheta^*)$ is different from zero.  Thus, we can follow this branch while the Jacobian $D\bm{F}(\hat{\bftheta}_{pr}(K))$ remains non-singular. By construction this stable solution verifies that  $\hat{\bftheta}_{pr}(K) \to \bf0$ as $K \to \infty$.

%More precisely, for all $K \in (K_{br}, \infty)$ the solutions of $\bfG(K,\hat{\bftheta})=\bf0$ are given by $\hat{\bftheta}=\hat{\bftheta}_{pr}(K)$.
%Thus, we can continue locally this solution as long as  the determinant of $\frac{\partial \bfG}{\partial \bftheta} (\bftheta^*)=K D\bm{F}(\bftheta^*)$ is different from zero. By hypothesis, we have $|D\bm{F}(\bftheta^*)| \neq 0$ and thus we can consider the map $K \mapsto \bftheta_{pr}^*(K)$.

We continue proving statement (c). We can now calculate the derivative along this principal branch $\bfG(K,\hat{\bftheta}_{pr}(K))=\hat{\bfomega} + K \bm{F} (\hat{\bftheta}_{pr}(K))=\bf0$ with respect to $K$. We obtain,

\begin{equation}\label{eq:derivative}
\frac{ \partial \hat{\bftheta}_{pr}(K)}{\partial K} = - \frac{1}{K} D\bm{F}(\hat{\bftheta}_{pr}(K))^{-1} \, F(\hat{\bftheta}_{pr}(K)).
\end{equation}

From Lemma~\ref{lem:DF} we know that $\bm{F}(\hat{\bftheta})=M \nabla H(\hat{\bftheta})$ and  $D\bm{F}(\hat{\bftheta})=M D^2H(\hat{\bftheta})$
where the auxiliary map  $H(\hat{\bftheta})$ is related to the order map $\hat{R}$ via  $H(\hat{\bftheta})=\frac{(n+1)^2}{2} \hat{R}(\hat{\bftheta})$.

We can finally compute the derivative of the one-variable map $K \mapsto H(\hat{\bftheta}_{pr}(K))$, where $\hat{\bftheta}_{pr}(K)$ is a stable solution of the Kuramoto model. We claim that $\frac{ d H(\bftheta^*(K))}{dK} >0$, thus proving that the order parameter map is strictly increasing. To see this claim, we combine the above computations and  Equation~\eqref{eq:derivative}:

\begin{eqnarray}
\displaystyle\frac{ d H(\hat{\bftheta}_{pr}(K))}{dK}  & = &  \nabla H(\hat{\bftheta}_{pr}(K))^T \cdot \frac{ \partial \hat{\bftheta}_{pr}(K)}{\partial K} \nonumber
\\
& = & -\frac{1}{K} \nabla H (\hat{\bftheta}_{pr}(K))^T D\bm{F}(\hat{\bftheta}_{pr}(K))^{-1} \, \bm{F}(\hat{\bftheta}_{pr}(K)) \nonumber
\\
& = & -\frac{1}{K}  \nabla H (\hat{\bftheta}_{pr}(K))^T \left[ M D^2H(\hat{\bftheta}_{pr}(K)) \right]^{-1} \, M \nabla H(\hat{\bftheta}_{pr}(K)) \nonumber
\\
& = & -\frac{1}{K}  \nabla H (\hat{\bftheta}_{pr}(K))^T  D^2H(\hat{\bftheta}_{pr}(K)) ^{-1}  \nabla H(\hat{\bftheta}_{pr}(K)) >0. \label{eq:derivative_order}
\end{eqnarray}
Since, by assumption, the linear map $D^2H(\bftheta^*(K))$ is negative definite, and this condition is equivalent to $u^T D^2 H(\bftheta^*(K)) u <0$ for all $u \in \mathbb R^n$ different from zero; and the inverse map of a negative definite map is also negative definite. We notice that the matrix $D^2H(\hat{\bftheta})$ is  symmetric, but the matrix $D\bm{F} (\hat{\bftheta})$ is not.

We finally prove statement (d).
As we mentioned above, we can continue the solution $K \mapsto \hat{\bftheta}_{pr}(K)$ until the determinant of $D\bm{F}(\hat{\bftheta}_{pr}(K))$ is equal to zero, or in other words, until
$\hat{\bftheta}_{pr}(K)$ belongs to the boundary $ \Omega$, since for all $\hat{\bftheta}$ in the boundary of  $\Omega$ we have that $\det(D\bm{F}(\hat{\bftheta}_{pr}(K)))=0$.  From Lemma~\ref{lem:DF}(b) we know that $D\bm{F}(\hat{\bftheta})=M D^2H(\hat{\bftheta})$. Hence, when $K \to K_{br}^+$, the determinant of $D^2H((\bftheta^*(K))$ tends to zero since $M$ is an invertible matrix. In Equation~\eqref{eq:derivative_order} we have computed the derivative of the map $K \mapsto H(\hat{\bftheta}_{pr}(K))$. We finally conclude that $\frac{ d H(\hat{\bftheta}_{pr}(K))}{dK} \to +\infty$,  since at least one eigenvalue of $\left[ D^2H((\bftheta^*(K))\right]^{-1}$ tends to infinity as $K \to K_{br}^+$.

%\frac{\partial \bfH}{\partial \theta_k}  & =   -\sin\theta_k
%\left(1+\displaystyle \sum_{j=1}^n \cos\theta_j \right)+  \cos(\theta_k) \left(\display style\sum_{j=1}^n \sin\theta_j\right) \\
%& =-\sin(\theta_k) + \displaystyle \sum_{j=1}^n [ \cos(\theta_k) \sin(\theta_j) -
%\sin(\theta_k) \cos(\theta_j)] \\
%a & b
 %\\& =-\sin(\theta_k) + \sum_{j=1}^n \sin(\theta_j - \theta_k)

%\[
%\bfH( \theta_1, \ldots, \theta_n) = \frac{1}{2} [1+
%\displaystyle \sum_{j=1}^n e^{i \theta_j}
%]^2
%\]

 %\end{proof}

\subsection{Geometric lemmas}

Theorem~B is a direct result from the following auxiliary Lemmas~\ref{lem:boundary_omega} and~\ref{lem:polytope_n}.
We denote by $\mathcal  C_n = \{ e_1, \ldots, e_n\}$  the canonical basis of $\mathbb R^n$  given by $e_i = (0, \ldots, 0, 1, 0, \ldots, 0)$ with entry 1 at the position $i$ for $1 \leq i \leq n$.  And we consider the $2n$~points given by
$p_i^{\pm} = \pm \frac{\pi}{2} e_i$. We denote these $2n$~points as the \emph{marked points}.

\begin{lemma}\label{lem:boundary_omega}
The $2n$ marked points $p_i^{\pm}$, for $1 \leq i \leq n$,  belong to the boundary of $\Omega$.
\end{lemma}

\begin{proof}
We just prove that the marked point $p_1^+=\frac{\pi}{2} e_1$ belongs to the boundary of $\Omega$. We consider
$s e_1$ with $s \in [0,\frac{\pi}{2}]$, and we can see that, on the one hand,  we have that $se_1$ belongs to  $\Omega$ for $s \in[0,\frac{\pi}{2})$ and, on the other hand, $\frac{\pi}{2} e_1$ belongs to the boundary of $ \Omega$. Easy computations show that  $D\bm{F}(s e_1)$ is given by
%\red{
%\[
%D\bm{F}(se_1) =
%\left(
%\begin{array}{ccccc}
%-(n+1)\cos(s) & \cos(s)-1 & \cos(s)-1 &  \ldots & \cos(s)-1  \\
%0  &   -(n+1)\cos(s) &  \cos(s)-1  & \ldots &  \cos(s)-1  \\
%0  & 0 & -(n+1)\cos(s) &   \ldots & \cos(s)-1 %\\
%\ldots & \ldots &  \ldots & \ldots & \ldots \\
%0 & 0 & 0 &  \ldots & -(n+1)\cos(s)
%\end{array}
%\right).
%\]
%}

%\red{This triangular matrix has a unique eigenvalue at $-(n+1)\cos(s)$ with multiplicity $n$. This eigenvalue is strictly negative for $s \in [0,\frac{\pi}{2})$ and equal to zero for $s=\frac{\pi}{2}$.}

\[
D\bm{F}(se_1) =
\left(
\begin{array}{ccccc}
-(n+1)\cos(s) & \cos(s)-1 & \cos(s)-1 &  \cdots & \cos(s)-1  \\
0  &   -n-\cos(s) &  0  & \cdots &  0  \\
0  & 0 & -n-\cos(s) & \cdots & 0 \\
\vdots & \vdots & \vdots & \ddots & \vdots \\
0 & 0 & 0 &  \cdots & -n-\cos(s)
\end{array}
\right).
\]

%\[
%D\bm{F}(se_1) =
%\left(
%\begin{array}{ccccc}
%-n\cos(s) & \cos(s) & \cos(s) &  %\ldots & \cos(s)  \\
%\cos(s)  &   -\cos(s)-(n-1) &  1  & \ldots &  1  \\
%\cos(s)  & 1 & -\cos(s)-(n-1) &   \ldots & 1 \\
%\ldots & \ldots &  \ldots & \ldots & \cos(s)-1\\
%\cos(s) & 1 & 1 &  \ldots & -\cos(s) - (n-1)
%\end{array}
%\right)
%\]

This triangular matrix has two unique eigenvalues, $-(n+1)\cos(s)$ with multiplicity $1$, and $-n-\cos(s)$ with multiplicity $n-1$. Both are negative for $s \in[0,\frac{\pi}{2})$, and only the first one becomes $0$ for $s=\frac{\pi}{2}$.
Hence, the point $p_1^+=\frac{\pi}{2} e_1$ belongs to the boundary of $\Omega$. For the remaining points $p_i^{\pm}$, the computations are equivalent.

%It is possible to compute explicitly the eigenvalues of the above matrix. There are three eigenvalues: $ -\cos(s)-n$ with multiplicity $n-2$, and two simple eigenvalues $\lambda_{\pm}$ given by
%\[
%\lambda_{\pm}= \frac{-((n+1)\cos(s)+1) \pm \sqrt{ ((n+1)\cos(s)+1)^2-4(\cos^2(s)+n\cos(s))} }{2}
%\]
%which are all negative for $s \in(0,\frac{\pi}{2})$. When $s=\frac{\pi}{2}$ the eigenvalues of $D\bm{F}(\frac{\pi}{2}e_1)$ are $0$, $-1$, and $-n$ with multiplicity $n-2$. So, the point $p_1^+$ belongs to the boundary of $\Omega$.
%For the remaining points $p_i^{\pm}$, the computations are equivalent.

\end{proof}

We denote by $\mathcal P_n$ the polytope in $\mathbb R^n$ with extremal $2n$ points $\bm{F}(p_i^{\pm})$. More precisely, $\mathcal P_n$ is the smallest convex set in $\mathbb R^n$ containing the $n$ segments $S_n$ given by
%\[
% S_n = \bigcup_{k=1}^n  \{t p_k^+  \, %| \,  t \in [-1,1] \}.
%\]

\[
 S_n = \bigcup_{k=1}^n  \{t \bm{F}(p_k^+)  \, | \,  t \in [-1,1] \}.
\]

%For example, if $n=2$, the polytope $\mathcal P_2$ contains the two segments $S_1 = \{ t (-2,1) \, | \,  t \in [-1,1]\}$  and $S_2=\{ t (1,-2) \, | \,  t \in [-1,1]\}$. We show in Figure~\ref{fig:marked_points} these segments and the Polytope $\mathcal P_2$.

\begin{figure}[tbp]
    \centering
%     \includegraphics[width=0.9\textwidth]{fig_3.pdf}
       % \put(-270,108){\small $y$}
       %           \put(-200,40){\small $x$}
       %           \put(-12,39){\small $y$}
       %           \put(-92,100){\small $z$}
       %           \put(-138,10){\small $x$}
       %           \put(-225,58){\small $\mathcal P_2$}
       %           \put(-65,83){\small $\mathcal P_3$}
       \includegraphics[width=0.49\textwidth]{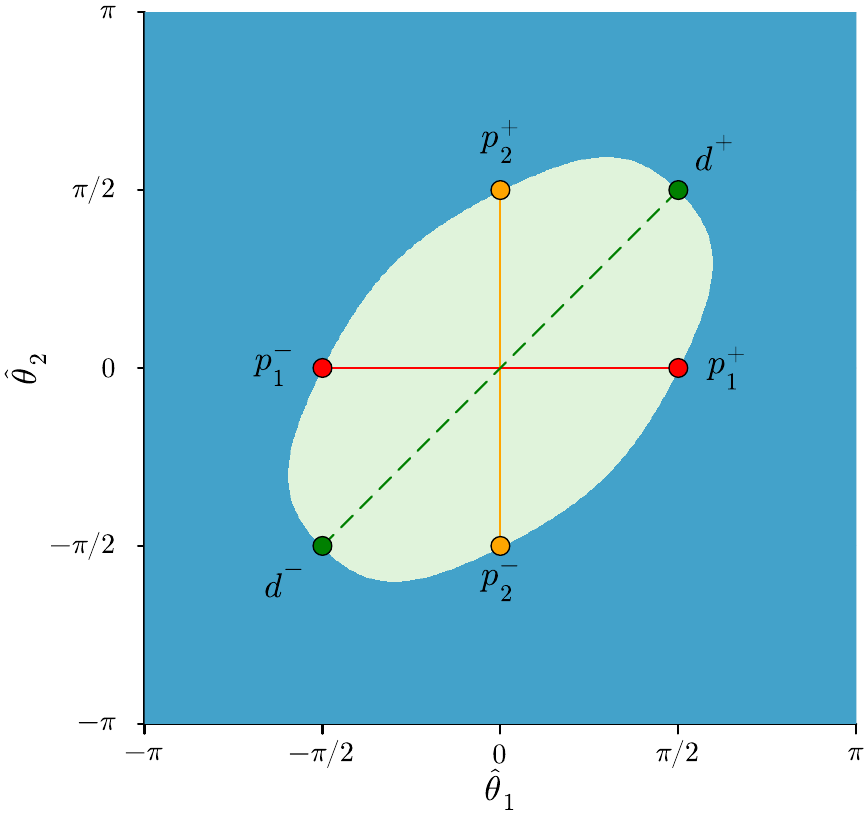}
       \includegraphics[width=0.47\textwidth]{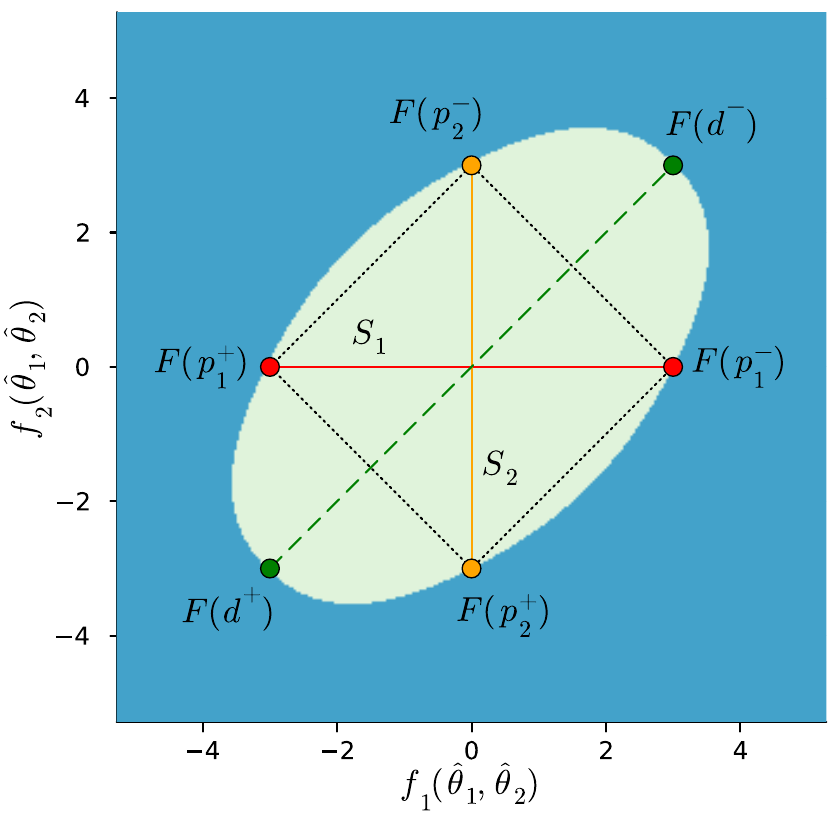}

\caption{\small
Stability region $\Omega\subset\mathbb T^2$ (left) and its convex image $\bm{F}(\Omega)\subset\mathbb R^2$ (right) for $n=2$, both in light green. The polytope $\mathcal P_2$ is marked with a dotted black line, together with the sets $S_1$ and $S_2$, their reference points $\bm{F}(p_1^\pm)$ and $\bm{F}(p_2^\pm)$, and the diagonal points $\bm{F}(d^\pm)$ (right).}
    \label{fig:marked_points}
    \end{figure}

\begin{lemma} \label{lem:polytope_n}
We denote by $\mathcal P_n$ the polytope with the extremal  points  $\bm{F}(p_i^{\pm})$ for $ 1 \leq i \leq n$. The polytope $\mathcal P_n$ has $2^n$ faces. Moreover,  the $2^n$ faces of the polytope  $\mathcal P_n$ are given by
\[
\varepsilon_1 x_1 + \varepsilon_2 x_2 + \ldots + \varepsilon_n x_n=n+1.
\]
\noindent where $\varepsilon_k \in \{ -1,1\}$ for all $1 \leq k \leq n$.

%containing the extremal points $\bm{F}(p_{i_1}^+), \ldots, \bm{F}(p_{i_k}^+), \bm{F}(p_{i_{k+1}}^-), \ldots, \bm{F}(p_{i_n}^-) $, where
%$i_1, \ldots, i_k$ are $k$ arbitrary different indexes in  the set  $ \{ 1, \ldots, n\}$  for any  $0 \leq k \leq n$   and $i_{k+1}, \ldots, i_{n}$ are the complementary indexes is given by the equation

%\[
 %(n-2k-1) ( x_{i_1} + x_{i_2} + \ldots + x_{i_k}) + (n-2k+1) ( x_{i_{k+1}}+ \ldots + x_{i_n})=n+1.
%\]
\end{lemma}

\begin{proof}

We can compute the images $\bm{F}(p_i^{\pm})$.  We easily have $\bm{F} (p_i^+) = -(n+1) e_i$ and $\bm{F} (p_i^-)  = (n+1)e_i $ for all $i=1, \ldots, n$. In order to find all the faces of the polytope $\mathcal P_n$  we only need to select $k$ points in the set $\{ \bm{F}(p_1^+), \ldots, \bm{F}(p_n^+)\}$ and the $n-k$ complement points in the set $\{  \bm{F}(p_1^-), \ldots, \bm{F}(p_n^-)\}$  where
$k$ runs from 0 to n.  As an example, the plane passing through the points $$\bm{F}(p_1^+), \ldots, \bm{F}(p_k^+), \bm{F}(p_{k+1}^-), \ldots, \bm{F}(p_n^-),$$
is the hyperplane  $$-x_1- \ldots - x_k + x_{k+1} + \ldots+ x_n =n+1.$$

In the general case, any of the $2^n$ faces of the Polytope contains a set of extremal points $$\bm{F}(p_{i_1}^+), \ldots, \bm{F}(p_{i_k}^+), \bm{F}(p_{i_{k+1}}^-), \ldots, \bm{F}(p_{i_n}^-), $$ where
$i_1, \ldots, i_k$ are arbitrary indexes in the set $ \{ 1, \ldots, n\}$  for any  $0 \leq k \leq n$, and $i_{k+1}, \ldots, i_{n}$ are the complementary indexes; in other words $i_{k+1}, \ldots, i_n \in \{1, \ldots, n \} \setminus \{ i_1, \ldots, i_k\}$. Thus, the hyperplane containing $\bm{F}(p_{i_1}^+), \ldots, \bm{F}(p_{i_k}^+), \bm{F}(p_{i_{k+1}}^-), \ldots, \bm{F}(p_{i_n}^-) $ is

\[
 -( x_{i_1} + x_{i_2} + \cdots + x_{i_k}) +  ( x_{i_{k+1}}+ \cdots + x_{i_n})=n+1
\]

\end{proof}

\subsection{Proof of Theorem B}\label{subsec:Proof_B}

%\begin{proof}[Proof of Theorem B]

We first claim  that $\bm{F}(\Omega)$ is a convex set. This result comes from \cite{Bronski_DeVille_Park} where the authors consider the full Kuramoto model \eqref{eq:main}. More precisely, Bronski, DeVille and Park proved that $\bm{f}(S_{\bftheta})$ is a convex set, where $S_{{\bftheta}}$ is the set of configurations $\bftheta$ for which the Jacobian $D\bff(\bftheta)$ is negative semi-definite with a one dimensional kernel.

We  denote by $\pi:\mathbb T^{n+1} \to \mathbb T^n$ quotient map
\[
\pi(\theta_1, \theta_2, \ldots, \theta_n, \theta_{n+1})=(\theta_1 - \theta_{n+1}, \ldots, \theta_{n}-\theta_{n+1})
\]
and the gauge section $\iota: \mathbb T^n \to \mathbb T^{n+1}$ defined by
\[
\iota(\hat{\theta}_1, \ldots, \hat{\theta
}_n)=(\hat{\theta}_1, \ldots, \hat{\theta}_n,0).
\]
We also define the linear map $L:\mathbb R^{n+1} \to \mathbb R^n$ defined by $$L(x_1, \ldots, x_n, x_{n+1})=(x_1-x_{n+1}, \ldots, x_n-x_{n+1})$$
Then, the reduced vector field is precisely
\[
\bm{F}(\hat{\bftheta})=L\bm{f}(\iota(\hat{\bf{\bftheta}})).
\]

The full Kuramoto vector field is invariant under uniform phase shifts, and
therefore $D\bm{f}(\boldsymbol\theta)\bm{1}=\bm{0}$. The condition that
$D\bm{f}(\boldsymbol\theta)$ is negative semidefinite with kernel
$\operatorname{span}\{\bm{1}\}$, as in
\cite{Bronski_DeVille_Park}, is equivalent to strict negative definiteness
of the induced Jacobian on the quotient by global phase shifts. In the above
gauge, this induced Jacobian is represented by $D\bm{F}$. Consequently,
\[
\Omega
=
\pi(S_\theta),
\]
or equivalently, $S_\theta$ is the $S^1$-orbit of $\iota(\Omega)$. Since $\bm{f}(S_\theta)$ is convex by
\cite{Bronski_DeVille_Park}, and since $L$ is linear, it follows that
\[
\bm{F}(\Omega)
=
L\bigl(\bm{f}(S_\theta)\bigr)
\]
is a convex set and the claim follows.

%\alex{We now clarify the relation between the reduced stability set used here and
%the stability set considered in \cite{Bronski_DeVille_Park}. Let
%$N=n+1$, and define the quotient map
%\[
%\pi(\theta_1,\ldots,\theta_N)
%=
%(\theta_1-\theta_N,\ldots,\theta_{N-1}-\theta_N),
%\]
%together with the gauge section
%\[
%\iota(\hat{\boldsymbol\theta})
%=
%(\hat\theta_1,\ldots,\hat\theta_n,0).
%\]
%We also define the linear map
%\[
%L:\mathbb R^{N}\to\mathbb R^n,\qquad
%L(v)_i=v_i-v_N .
%\]
%Then the reduced vector field is precisely
%\[
%\bm{F}(\hat{\boldsymbol\theta})
%=
%L\,\bm{f}(\iota(\hat{\boldsymbol\theta})).
%\]

%The full Kuramoto vector field is invariant under uniform phase shifts, and
%therefore $D\bm{f}(\boldsymbol\theta)\bm{1}=0$. The condition that
%$D\bm{f}(\boldsymbol\theta)$ is negative semidefinite with kernel
%$\operatorname{span}\{\bm{1}\}$, as in
%\cite{Bronski_DeVille_Park}, is equivalent to strict negative definiteness
%of the induced Jacobian on the quotient by global phase shifts. In the above
%gauge this induced Jacobian is represented by $D\bm{F}$. Consequently,
%\[
%\Omega
%=
%\pi(S_\theta),
%\]
%or equivalently $S_\theta$ is the $S^1$-orbit of $\iota(\Omega)$.

%Since $\bm{f}(S_\theta)$ is convex by
%\cite{Bronski_DeVille_Park}, and since $L$ is linear, it follows that
%\[
%\bm{F}(\Omega)
%=
%L\bigl(\bm{f}(S_\theta)\bigr)
%\]
%is convex.}

 Let $\hat{\bfomega} \neq \bf0$ be a vector of frequencies. By construction, the polytope $\mathcal P_n$ is a subset of $\bm{F}(\Omega)$. We consider the ray $t\hat{\bfomega}$ for $t >0$. We first notice that for $t$ sufficiently small we have that $t \hat{\bfomega}$ belongs to $ \mathcal P_n \subset \bm{F}(\Omega)$. This ray first intersects the boundary of the polytope $\mathcal P_n$,  and later on this ray it intersects the boundary of $\bm{F}(\Omega)$, since $\mathcal P_n \subset \bm{F}(\Omega)$. We denote by $K_b$ and $K_{br}$ the inverse of the values of $t$ such that $\frac{\hat{\bfomega}}{K_b} \in \partial \mathcal P_n$ and $\frac{\hat{\bfomega}}{K_{br}} \in \partial \bm{F}(\Omega)$, respectively, thus obtaining $K_b \geq K_{br}$.

%We first compute the coordinates of $\hat{\bfomega}$ in the basis $\mathcal B$ (see Lemma~\ref{lem:faces}).
The sign of the coordinates of $\hat{\bfomega}$ determines the face where the ray $t \hat{\bfomega}$ first  intersects the boundary of the polytope $\mathcal P_n$. In order to construct this face we select $n$ points in this hyperplane. Concretely, we select $n$ points in the set $\{ \bm{F}(p_1^{\pm}), \ldots, \bm{F}(p_n^{\pm})$\}. If $\hat{\omega}_j >0$ we select the point $\bm{F} (p_j^-)$ since it is the position vector in the canonical basis, and when $\hat{\omega}_j \leq0$ we select $\bm{F}(p_j^+)$. The rationale for this is that, the hyperplane that is first hit is such that vector $t\hat{\bfomega}$ is a linear combination of the position vectors with non-negative coefficients, thus we select the position vectors so as this condition is satisfied. With this procedure, we have selected $k$ indexes $i_1, \ldots, i_k$ corresponding to $\bm{F}(p_{i_1}^+), \ldots, \bm{F}(p_{i_k}^+)$ and $n-k$ indexes $i_{k+1}, \ldots, i_n$ corresponding to $\bm{F}(p_{i_{k+1}}^-), \ldots, \bm{F}(p_{i_n}^-)$. From Lemma~\ref{lem:polytope_n}, the selected face is given by
\[
 - ( x_{i_1} + x_{i_2} + \ldots + x_{i_k}) +  ( x_{i_{k+1}}+ \ldots + x_{i_n})=n+1.
\]
Finally, the intersection between the ray $t\hat{\bfomega}$ and the above hyperplane is given at the point $t_b$ given by
\[
t_b = \frac{n+1}{ - ( \hat{\omega}_{i_1} + \hat{\omega}_{i_2} + \ldots + \hat{\omega}_{i_k}) +  ( \hat{\omega}_{i_{k+1}}+ \ldots + \hat{\omega}_{i_n})} = \frac{n+1}{\displaystyle \sum_{j=1}^n |\hat{\omega}_j|}.
\]
Hence
\[
K_b = \frac{1}{t_b}=  \frac{1}{n+1} \, \displaystyle \sum_{j=1}^n |\hat{\omega}_j|.
\]
is an upper bound of the critical parameter $K_{\ell}$, since $K_b \geq K_{br} \geq K_{\ell}$.

%\end{proof}

\subsection{Improving the upper bound $K_b$}

We briefly explain how to improve the value of the upper bound $K_b$ of $K_{\ell}$ obtained in Theorem~B. The geometric idea behind this upper bound is to approximate $\bm{F}(\Omega)$ by a polytope $\mathcal P_n$ constructed from $2n$ points on the boundary of $\bm{F}(\Omega)$. We can improve the bound by adding to the polytope more vertices laying at the boundary of $\bm{F}(\Omega)$.

Let us consider the two diagonal points $d^{\pm}= \pm\frac{\pi}{2}(1, \ldots, 1)$. It is easy to see that both diagonal points belong to $\partial \Omega$. Their images are given by
$\bm{F}(d^{+})=-(n+1)(1, \ldots, 1)$ and $\bm{F}(d^{-})=(n+1)(1,\ldots,1)$. Moreover, we notice that
\[
\begin{array}{l}
\bm{F}(d^{-}) = \bm{F}(p_1^-) + \ldots + \bm{F}(p_n^-)  \\
\bm{F}(d^{+}) = \bm{F}(p_1^+) + \ldots + \bm{F}(p_n^+).
\end{array}
\]
%since
%\[
%\displaystyle \sum_{i=1}^n u_i = (n-(n-1), n-(n-1), \ldots, n-(n-1))=(1, \ldots, 1)=\bm{F} (d^{-}).
%\]

We can refine the polytope $\mathcal P_n$ adding these two diagonal points $\bm{F}(d^{\pm})$.
Let us assume that $\hat{\bfomega}=(\hat{\omega}_1, \ldots, \hat{\omega}_n)$  with  all $\hat{\omega}_j>0$, or in other words, we have selected the face of the polytope $\mathcal P_n$ containing the points $\bm{F}(p_1^-),\bm{F}(p_2^-), \ldots, \bm{F}(p_n^-)$. In this case, we can use the point $\bm{F}(d^{-})$ to construct $n$ extra hyperplanes to obtain a better approximation of the intersection point between the ray $t \hat{\bfomega}$ and the boundary of the set $\bm{F}(\Omega)$. We briefly sketch how to do this refinement. These $n$ extra hyperplanes are defined as the hyperplanes containing $n-1$ points in the set $\{\bm{F}(p_1^-),\bm{F}(p_2^-), \ldots, \bm{F}(p_n^-)\}$ and the point $\bm{F}(d^-)$. We can compute explicitly these $n$-hyperplanes: the equation of the plane containing the $n$ points $\{\bm{F}(p_1^-), \ldots, \bm{F}(p_{s-1}^-), d^-,\bm{F}(p_{s+1}^-), \ldots,\bm{F}(p_{n}^-) \}$ is given by
\[
(2-n) x_s  +  \sum_{i=1, i \neq s}^n x_i = n+1.
\]
%It has been obtained by imposing that a general hyperplane $a_1 x_1 + a_2 x_2 + \cdots + a_n x_n = 1$ passes through the desired $n$ points.

In the symmetric case, the hyperplane passing through the $n$ points \newline $\{\bm{F}(p_1^+), \ldots, \bm{F}(p_{s-1}^+), d^+,\bm{F}(p_{s+1}^+), \ldots,\bm{F}(p_{n}^+) \}$  is given by
\[
(2-n)x_s  + \sum_{i=1, i \neq s}^n x_i = -(n+1).
\]

%For this $\hat{\bfomega}$, Lemma~\ref{lem:faces} tells us that $\lambda_j >0$ for all $1\leq j \leq n$.
We denote $\hat{\omega}_s = \displaystyle \min_{1 \leq j \leq n} \{\hat{\omega}_j\}$, or in other words, let $s$ be the index where this minimum is attained. Thus, we have
\[
\begin{array}{ll}
\hat{\bfomega} = & \hat{\omega_{1}} e_1 + \ldots + \hat{\omega_{n}} e_n = \displaystyle\sum_{i=1, i \neq s}^n (\hat{\omega}_i -\hat{\omega}_s) e_i + \hat{\omega}_s (e_1 + \cdots + e_n) \\
 & =\displaystyle\sum_{i=1, i \neq s}^n (\hat{\omega}_i -\hat{\omega}_s) e_i + \hat{\omega}_s {\bf0} \bm{F}( d^-)
\end{array}
\]
The value of $s$ corresponds to the index of the hyperplane where the ray $t \hat{\bfomega}$ hits first these $n$ extra hyperplanes. Hence, we have obtained the value $t^1_b$ corresponding to the intersection between the ray $t \hat{\bfomega}$ and the hyperplane $(2-n) x_s  +  \displaystyle\sum_{i=1, i \neq s}^n x_i  = n+1$. Finally,
\begin{equation}\label{eq:K_b1}
K^1_b = \frac{1}{t^1_b} = \frac{1}{n+1} \left[(2-n) \hat{\omega}_s  +  \displaystyle \sum_{i=1, i \neq s}^n \hat{\omega}_i\right] = \frac{1}{n+1} \left[(1-n) \hat{\omega}_s  +  \displaystyle \sum_{i=1}^n \hat{\omega}_i\right]
\end{equation}
%=\frac{1}{n+1} \left[(1-n) \hat{\omega}_s  + 2 \displaystyle \sum_{i=1}^n \hat{\omega}_i\right].

By construction, $K_{\ell} \leq K_b^1\leq K_b$, i.e., $K_b^1$ is a better upper bound of $K_{\ell}$. We can check that $K^1_b \leq K_b$. Using the expression of $K_b=\frac{1}{n-1} \displaystyle \sum_{j=1}^n \hat{\omega}_j$ since all $\hat{\omega}_j$ are positive. Therefore,
\[
K_b^1 \leq K_b
\Leftrightarrow
(1-n) \hat{\omega}_s +  \displaystyle\sum_{j=1}^n \hat{\omega}_j \leq
 \displaystyle\sum_{j=1}^n \hat{\omega}_j
\Leftrightarrow
0 \leq (n-1) \hat{\omega}_s,
\]
which is positive since, by assumption, $\hat{\omega}_j >0$ for $1 \leq j \leq n$.

Similarly, if $\hat{\bfomega}$ is such that we have selected the face of the polytope $\mathcal P_n$ containing the points $\bm{F}(p_1^+), \ldots, \bm{F}(p_n^+)$, i.e., $\hat{\omega}_j < 0$ for all $1\leq j \leq n$. In this case we add the diagonal point $\bm{F}(d^+)$ and denoting $\hat{\omega}_r = \displaystyle \max_{1 \leq j \leq n} \{\hat{\omega}_j\}$, the bound is now
\begin{equation}\label{eq:K_b2}
K^2_b = \frac{-1}{n+1} \left[(2-n) \hat{\omega}_r +  \displaystyle \sum_{j=1, j \neq r}^n \hat{\omega}_j\right] =
\frac{1}{n+1} \left[(1-n) |\hat{\omega}_r|  +  \displaystyle \sum_{j=1}^n |\hat{\omega}_j|\right].
\end{equation}

This refinement of the upper bound of $K_{\ell}$, making use of the diagonal points $d^{\pm}$, only applies to the cases in which $\hat{\bfomega}$ is such that all $\hat{\omega}_j$ are positive or all are negative, which limits its reach. However, this methodology could be extended in a similar way by adding other points on the boundary of $\bm{F}(\Omega)$, e.g., those of the form $\bm{F} (\pm\frac{\pi}{2}(a_1,\ldots,a_n))$ with $a_i\in\{0,1\}$ for all $1\leq i \leq n$, thus enabling a much richer applicability.

The specific improvement of the bound here detailed is especially important for frequency vectors $\hat{\bfomega}$ close to the diagonal, i.e., with all values around a center frequency and with small deviations. For example, if we take frequencies normally distributed with mean~$2$ and standard deviation~$0.05$, $\hat{\omega}_i\sim N(2, 0.05^2)$ for all $1\leq i\leq n$, in one of the realizations we obtain $K_{\ell}=0.02$, $K_b=1.98$, and $K_b^1=0.17$, which shows a drastic enhancement of the bound.

\subsection{Analysis for frequencies at the vertices of the polytope}

We have shown that the vertices $\bm{F}(p_i^{\pm})$ of the polytope $\mathcal P_n$ lay at the boundary of the stable region $\bm{F} (\Omega)$. This means that, in these vertices, the bound on the fully phase-locked state becomes exact, i.e., $K_b = K_{\ell}$; in fact, we have $K_b = K_{\ell} = 1$ since $t_b=1$.

Let us consider one of the vertices, e.g., $\bm{F}(p_1^{+})$, which corresponds to a vector of frequencies $\hat{\bfomega}=-\bm{F}(p_1^{+})=(n+1,0,\ldots,0)$ for $K=1$, and a stationary solution $\bftheta^*=\frac{\pi}{2}e_1$. For $K\ge K_{\ell}$, all phases are constant at the stationary state, $\dot{\hat{\theta}}_1=0$, but we do not know which $\bftheta^*$ is solution of $\hat{\bfomega}+K\bm{F}(\bftheta^*)=0$. An ansatz could be $\bftheta^*=(\hat{\theta}_1,0,\ldots,0)$, which is also proportional to $e_1$. Substituting this ansatz into Eqs.~\eqref{eq:main1} and~\eqref{eq:f_hat} in the stationary state, we get:
\[
\left\{
\begin{array}{l}
    \displaystyle 0 = (n + 1) + K (-\sin \hat{\theta}_1 - \sin \hat{\theta}_1 + (n-1)\sin(0-\hat{\theta}_1)) = (n + 1) (1 - K\sin \hat{\theta}_1)
    \\
    \displaystyle 0 = 0 + K (-0 - \sin \hat{\theta}_1 + \sin(\hat{\theta}_1-0)) = 0
\end{array}
\right.
\]
The solution is as follows:
\[
  \sin\hat{\theta}_1=\frac{1}{K} \ \ \Longrightarrow \ \ \bftheta^* = (\arcsin K^{-1},0,\ldots,0)
\]

Since there are only two different frequencies, we can write the order parameter~\eqref{eq:R_red} as:
\[
  \hat{R}=\frac{1}{n+1}\left| 1+e^{i\hat{\theta}_1}+(n-1) \right|
  =
  \frac{1}{n+1}\sqrt{1+n^2+2n\cos \hat{\theta}_1}
\]
Thus, the analytic value of the order parameter for $K\ge K_{\ell}=1$ is:
\begin{equation}\label{eq:R_polytope_vertex}
    \hat{R}_{\text{theo}}(K)=\frac{n}{n+1}\sqrt{1+\frac{2\sqrt{K^2-1}}{Kn}+\frac{1}{n^2}}.
\end{equation}

\begin{figure}[tb]
    \centering
     \includegraphics[width=0.9\textwidth]{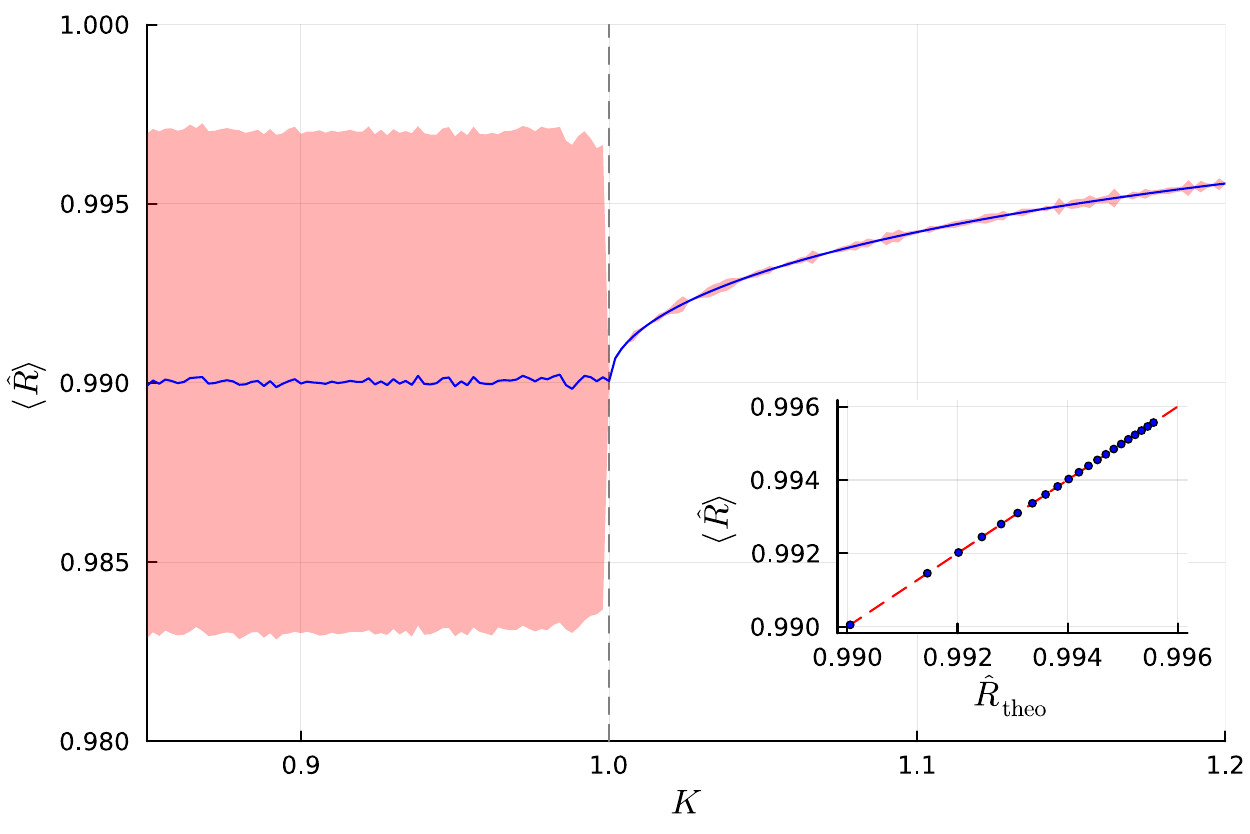}
           \caption{\small{
Zoom of the phase diagram of the Kuramoto model for $n+1=100$} and with a vector of frequencies laying at one of the vertices of the polytope, namely, $\hat{\bfomega}=\bm{F}(p_1^{-})=(n+1,0,\ldots,0)$. Similarly to Figure~\ref{fig:order_parameter}, we represent $\langle \hat{R} \rangle$ as a blue solid line, surrounded by a pale-red ribbon of width $\sigma_{\hat{R}}$. The vertical dashed line at $K=K_{\ell}=K_b=1$ marks the transition to the fully phased-locked transition. The inset shows the agreement between the values of $\langle \hat{R} \rangle$ and those calculated analytically, $\hat{R}_{\text{theo}}$, using Eq.~\eqref{eq:R_polytope_vertex}, for~21 values of $K$ equally spaced between $K=1.0$ and $K=1.2$.}
      \label{fig:R_polytope_vertex}
\end{figure}

In Figure~\ref{fig:R_polytope_vertex} we show the behavior of $\hat{R}(K)$ close to the critical threshold $K_{\ell}=1$, and that Eq.~\eqref{eq:R_polytope_vertex} matches perfectly with the experimental values for $K\ge K_{\ell}$.

\section{Discussion}

Determining the coupling strength required for synchronization is a central problem in the theory of coupled oscillators. In finite populations, and in the mean-field all-to-all setting, the onset of collective behavior can be characterized sharply as the loss of stability of the incoherent state. By contrast, a more stringent and structurally richer question concerns the existence of stable fully phase-locked equilibria in an appropriate co-rotating frame. This reduces to a nonlinear fixed-point problem for the phase differences, whose transcendental character makes closed-form solutions exceptional rather than typical.
A useful viewpoint, developed in related contexts, is that stability defines a region of phase space and that the vector field maps this region into a structured subset of frequency space \cite{Bronski_DeVille_Park}.
In this paper, we exploit this convex-geometric structure to obtain a fully explicit finite-size bound on the critical coupling: we construct an analytically tractable polytope $\mathcal P_n\subset \bm{F}(\Omega)$ from distinguished boundary points of the stability region and compute its faces in closed form.
Intersecting the ray $t\hat{\bfomega}$ with $\partial\mathcal P_n$ yields an explicit analytical upper bound $K_b \geq K_{\ell}$ that is directly computable from the frequency vector and captures the finite-size nature of the transition.
This approach is complementary to classical mean-field estimates and spectral criteria, in that it provides a deterministic, non-asymptotic bound tailored to a given realization of frequencies.
Beyond the explicit bound, the geometric formulation clarifies the intrinsically finite-size nature of the fully phase-locked stability transition and ties its onset to the loss of strict negative definiteness (degeneracy) of the stability Jacobian along the locked branch. This convex viewpoint provides a natural route to systematically refine outer approximations of $\bm{F}(\Omega)$, to analyze alternative coupling architectures, and to extend geometric synchronization criteria beyond the all-to-all setting.

\section*{Acknowledgements}
The authors acknowledge support from Generalitat de Catalunya 2021SGR-633, and Universitat Rovira i Virgili 2023PFR-URV-00633.
A.A.\ and S.G.\ acknowledge support from MICIN PID2021-128005NB-C21 and RED2022-134890-T, and Universitat Rovira i Virgili 2025INTER-03 (ComSCIAM).
A.A. also acknowledges the ICREA Academia program from Generalitat de Catalunya.

\bibliographystyle{alpha}
\bibliography{references}

\end{document}